\documentclass[%
 reprint, amsmath, amssymb, aps, superscriptaddress, prx]{revtex4-2}
\usepackage{graphicx}% Include figure files
\usepackage{multirow}
\usepackage{dcolumn}
\usepackage{booktabs}
\usepackage{siunitx}
\usepackage{bm}% bold math
\usepackage{hyperref}% add hypertext capabilities
\usepackage[mathlines]{lineno}% Enable numbering of text and display math
\usepackage{physics}
\usepackage{float}
\usepackage{diagbox}
\usepackage{booktabs}
\usepackage[normalem]{ulem}
\usepackage{resizegather}
\usepackage{newtxmath}

\usepackage{amsmath, amssymb, amsthm}
\newtheorem{theorem}{Theorem}

\newtheorem{definition}{Definition}
\usepackage{algorithm}
\usepackage{algpseudocode}

\hypersetup{
    colorlinks=true,
    linkcolor=red,
    filecolor=magenta,
    citecolor=blue,
    urlcolor=blue,
}

\newcommand{\IITB}{Department of Physics, Indian Institute of Technology Bombay, Powai, Mumbai, Maharashtra 400076, India}
\newcommand{\Quicst}{Centre of Excellence in Quantum Information, Computation, Science and Technology,
Indian Institute of Technology Bombay, Powai, Mumbai, Maharashtra 400076, India}

\begin{document}

\title{Efficient Pauli-decomposition and multistage state-refinement for tensor network based differential equation solver}
\author{Vishwabhushan Suresh Gholap} \affiliation{\IITB}

\author{Himadri Shekhar Dhar} \affiliation{\IITB}\affiliation{\Quicst}

\author{Siddhartha Santra} \affiliation{\IITB}\affiliation{\Quicst}

\date{\today}% It is always \today, today,

\begin{abstract}

\label{abstract}
Classical numerical techniques for solving partial differential equations become computationally expensive as the dimension of the discretized matrix representation of the underlying differential operator increases. For differential operators that give rise to a Sturm--Liouville problem (finding the eigensystem of a linear, second-order, Hermitian differential operator in one-dimension), tensor network based methods can be highly productive. The primary advantage lies in the fact that an operator of dimension $N\times N$ can be efficiently represented via a matrix product operator (MPO) using only  $n=\log_2{(N)}$ qubits, which opens up the possibility of finding the eigenvalues and eigenvectors using methods such as imaginary time evolution. However, in practice, this can be computationally challenging. First, most known methods to generate MPOs of large operators without a manifest tensor product structure, require prohibitively large memory, and second, the number of Trotterization steps necessary for convergence in conventional imaginary time evolution increases rapidly with $n$.  
In this work, we present techniques to mitigate these two challenges for certain sparse and structured differential operators. To address the first challenge, we construct the MPO corresponding to the differential operator by efficiently expanding it in the basis of Pauli-operator strings, followed by iterative compression into a MPO with the desired bond dimension. The efficiency of this step is enabled by an analytical expression we derive for the Pauli basis coefficients by exploiting their binary encoding, reducing the memory requirement from $\mathcal{O}(2^{n+1})$ to $\mathcal{O}(2n)$. To address the second challenge, we propose a multistage state-refinement heuristic that accelerates the convergence of the imaginary time evolution, reducing the time to convergence by up to two orders of magnitude already for systems with as few as ten qubits. Using the improved tensor-network framework, we compute the first 32 eigenstates of a Laplacian with a dimension larger than $10^6$ with fidelity above $0.95$, using a $20$ qubit MPO. Furthermore, we validate our method on the two-dimensional anharmonic oscillator and investigate its performance for disordered systems, where increasing random-potential strength degrades accuracy and ultimately limits the applicability of the approach.
\end{abstract}

\maketitle

\section{Introduction}
\label{sec:intro}
Partial differential equations play a central role in modeling systems and dynamics across all fields of science and engineering. However, analytical techniques for solving PDEs such as separation of variables, Green's function techniques and integral transforms yield closed form solutions only for restricted geometries and boundary conditions, necessitating the use of numerical techniques, such as finite element methods, for more complex systems. Since these numerical techniques discretize the solution domain into a finite set of grid points, finite-element methods typically have computational costs that scale at least linearly with the number of grid points. As the grid becomes finer and the problem size grows, for certain classes of PDEs, tensor network (TN) methods \cite{White1992,White1993,Schollwock2011,Orus2014,Verstraete2008,Cirac2021,Orus2019,Garcia2024} can provide a more efficient alternative approach to conventional numerical techniques by exploiting underlying low-rank structure in the solution \cite{Eisert2010,Hastings2007,Vidal2003,Vidal2004}.

The tensor network approach is particularly well suited to linear, second-order PDEs that, upon separation of variables, yield a set of Sturm–Liouville ordinary differential equations (ODEs). In this case, the differential operators associated with the ODEs together with the boundary conditions define linear, second-order, Hermitian differential operators that can be discretized and mapped onto the dynamics of an appropriate quantum system~\cite{Costa2019,Jin2023}. Specifically, the discretized ODE can be recast as a one-dimensional quantum lattice model, with the solution vector encoded as a matrix product state (MPS) and the corresponding operator represented as a matrix product operator (MPO) \cite{Oseledets2011,OseledetsDolgov2012,Khoromskij2011,Kazeev2013,Dolgov2014,DolgovEig2014,Holtz2012,Boelens2018}. The resulting tensor network representation enables the application of algorithms such as the density matrix renormalization group (DMRG)~\cite{White1992, White1993,Schollwock2011}, time-evolving block decimation (TEBD)~\cite{Vidal2003}, or imaginary-time evolution techniques~\cite{Daley2004,Vidal2007,Haegeman2011,Paeckel2019} 
to determine the eigenvalues and eigenfunctions of the differential operator which can provide a complete basis for representing the solutions to the ODE for the given boundary conditions.

The efficacy of the tensor network approach for computing the eigensystem of a discretized differential operator depends on both the computational cost of constructing an efficient MPO representation of the operator and the cost of determining its eigensystem using a suitable tensor network algorithm, which in the present work is imaginary time evolution. However, there are challenges in the efficiency of both these steps. First, arbitrary discretized differential operators do not generally possess manifestly compact low-rank MPO representations. As a result, the construction of an MPO can itself become a significant space- and time-complexity bottleneck, since the large matrix associated with a fine-grid discretization must typically be formed and stored explicitly prior to its compression into MPO form. In addition, the compression procedure itself requires repeated singular-value decompositions of large matrices to obtain the MPO representation, further increasing the computational cost. Second, imaginary-time evolution, which relies on iteratively applying Suzuki–Trotter decomposed form of the exponentiated MPO, may require a large number of Trotter steps for the large matrices associated with fine-grid discretizations, potentially slowing convergence to the target eigenstates.

In this work, we present techniques to mitigate both of these challenges, thereby potentially improving the performance of tensor network methods for solving PDEs, particularly when the discretized differential operators are sparse and structured. First, we exploit the observation that many operators of practical interest admit a direct MPO construction through a Pauli-basis decomposition, avoiding the need to explicitly construct and store the full operator matrix. Moreover, this decomposition naturally enables the application of Suzuki–Trotter expansions, which are central to time evolution algorithms used to compute the operator's eigenstates.

The use of a Pauli-basis decomposition, however, introduces its own challenges. In particular, the number of Pauli strings appearing in the expansion can be as large as $4^n$ for an operator with no underlying symmetries, and this exponential scaling is, in general, unavoidable. Nevertheless, as we show, the evaluation of the corresponding expansion coefficients can be made highly efficient. Pauli strings are naturally represented as tensor products rather than explicit matrices. Existing approaches typically rely on sparse-array representations and iterative procedures to construct the matrix representation of each Pauli string, followed by matrix multiplication with the operator and evaluation of the resulting trace to obtain the corresponding expansion coefficient \cite{Koska2024}. To avoid both the memory-intensive construction of Pauli-string matrices and the computationally expensive matrix multiplications, we derive a closed-form analytical expression for the matrix elements of arbitrary Pauli strings by exploiting their binary encoding \cite{gottesman1997stabilizercodesquantumerror}. This allows the matrix elements of Pauli strings to be evaluated directly, without explicitly constructing their matrix representations. The resulting expressions enable the efficient evaluation of traces between operators with analytically known matrix elements and arbitrary Pauli strings. By eliminating both the iterative construction of Pauli-string matrices and the associated large-scale matrix multiplications, the computational cost of determining each expansion coefficient is substantially reduced, leading to a significantly more efficient MPO construction procedure for the discretized differential operator.

Second, to mitigate the challenge from the growth in the number of Suzuki–Trotter steps required to achieve convergence as the grid resolution increases, we introduce a multistage approach. First, imaginary-time evolution is performed on a coarse grid to obtain approximations to the ground state and several low-lying excited states. These states are then interpolated onto a finer grid and used to initialize a second imaginary-time evolution calculation. This works because the low-lying eigenstates of differential operators are typically smooth and do not oscillate too fast, implying that interpolated coarse-grid eigenstates provide good approximations to the corresponding fine-grid eigenstates. Consequently, imaginary-time evolution on the finer grid is initialized much closer to the target state, substantially reducing the computational cost. Thus multistage state-refinement heuristic significantly accelerates convergence and reduces the number of Trotter steps required on finer discretizations.

To test our improvised methods, we first considered the Laplacian as the differential operator, with boundary conditions permitting analytical solutions which were used to benchmark the eigenstates obtained via our method. For operators with dimension larger than $10^6$, our method can reproduce the first 32 eigenstates with fidelity greater than 0.95, using an MPO constructed using $20$ qubits. We then applied the method to systems without analytical solutions, such as the 2D anharmonic oscillator which gives rise to two one-dimensional Sturm–Liouville problems.Using an appropriately coarsened discretization to obtain a matrix representation of manageable size, we validated our method by comparing its results with the exact diagonalization of the corresponding discretized operator. Finally, we investigated disordered systems, where the differential operator is given by the sum of the Laplacian and a position-dependent random potential. We observed that as the strength of the potential increases, the fidelity of the solution vectors decreases and the error in the eigenvalues increases. Beyond a certain potential strength, the method ceases to work for the random potential. Tensor network methods are effective when the system possesses symmetries and low entanglement. Disordered or random potentials increase entanglement, making singular value decomposition truncation ineffective and increasing the required bond dimension.

Recent progress in quantum and quantum-inspired PDE solvers includes quantum linear system algorithms (e.g., HHL), high-order methods for differential equations, Hamiltonian simulation, and operator-decomposition techniques \cite{Harrow2009,Childs2017,Berry2014,Childs2020,Farghadan2025,Arseniev2024,Koska2024}, 
as well as variational and hybrid quantum-classical approaches for linear and nonlinear problems and eigenstate preparation \cite{Peruzzo2014,McClean2016,Cerezo2021,Lubasch2020,McArdle2019,Motta2020}. 
At the same time, tensor network methods have emerged as a powerful classical framework for large-scale matrix manipulations \cite{Oseledets2011,OseledetsDolgov2012,Khoromskij2011,Kazeev2013,Dolgov2014,Holtz2012,Boelens2018,Bachmayr2023} which exploit low-rank structures and have been successfully applied to high-dimensional PDEs using alternating least-squares optimization, DMRG-inspired methods, and related tensor network techniques \cite{Dolgov2014,Schollwock2011,Paeckel2019}.

The link between quantum algorithms and tensor networks is increasingly clear: tensor networks can act as both classical simulators of quantum systems and quantum-inspired solvers exploiting entanglement structure, variational forms, and operator factorizations \cite{Orus2014,Cirac2021,Orus2019,Garcia2024,Verstraete2008}, implying that some advantages of quantum PDE solvers can already be achieved classically for low-entanglement states and operators \cite{Eisert2010,Hastings2007}.

The practical utility of the methods proposed in this work depends on the sparsity and structure of the discretized operator. In particular, the MPO creation approach using a Pauli-basis decomposition works when the operator can be decomposed into a polynomial number of Pauli strings while the multistage state-refinement works when the eigenstates are not too oscillatory. 
As such the contribution of the work is not necessarily in achieving overall superior performance over other techniques for solving PDEs, rather it is to highlight efficient primitives, which may be of independent interest, that can accelerate Pauli decomposition and imaginary time evolution for a wide class of sparse operators including differential operator problems.
 
The paper is arranged as follows. A detailed explanation of our method is presented in the subsections of Sec.~\ref{sec:formalism_and_methods}. Subsection~\ref{sec:tn} gives a basic introduction to tensor network methods, \ref{sec:pauli} explores Pauli-basis decomposition using binary representation,  \ref{sec:tridiag} specializes Pauli-basis decomposition for a tridiagonal matrices, in \ref{sec:ODE} we discuss the creation of matrix product operator and Suzuki--Trotter expansion to implement imaginary time evolution and in \ref{sec:Multistage state refinement} we discuss the multistage state-refinement heuristic.
In Sec.~\ref{sec:results}, we test our method on separable partial differential equations which reduce to Sturm--Liouville problems. We first benchmark on the one-dimensional diffusion equation by computing eigenstates of the Laplacian operator \ref{res:diff}. We then apply our method to a two-dimensional anharmonic oscillator \ref{res:anho}. Finally, we examine one-dimensional systems with quadratic, quartic, and random potentials to analyse fidelity and relative energy error variations with change in bond dimension and potential strength \ref{res:random}. We also present a way to estimate maximum error in the construction of matrix product operator representations \ref{res:err}. We end with a conclusion in Sec.~\ref{sec:conclusion}.

\section{Formalism and methods}
\label{sec:formalism_and_methods}

\subsection{Tensor network representation}
\label{sec:tn}

Tensor networks are structured representations of complex quantum states and operators in terms of a simple collection or network of connected tensors, which allows for
efficient storage and manipulation. 
For example, consider an $n$-qubit quantum state
\begin{equation*}
    \ket{\psi} =
    \sum_{i_1,\dots,i_n=0}^1
    \psi_{i_1 \dots i_n}
    \ket{i_1 \dots i_n},
\end{equation*} 
{where the coefficient tensor
$\psi_{i_1 \dots i_n}$ 
contains $2^n$ complex entries. 
For large $n$ storing it becomes infeasible.} 
The central idea of a tensor network is to factorize this large tensor into a sequence of smaller tensors connected by internal indices. The dimensions of these internal indices,
called bond dimensions, are governed by the entanglement structure of the state and determine the efficiency of the representation~\cite{Orus2014,Orus2019}. Whenever the tensors admit approximate low-rank factorizations across successive bipartitions, such a representation becomes efficient.
{For instance, using successive singular value decomposition (SVD) the tensor $\psi_{i_1 \dots i_n}$ can be represented as a network of $n$ tensors
to form a matrix product state (MPS)~\cite{Vidal2003, Schollwock2011, Orus2014}} 
\begin{equation*}
    \psi_{i_1 \dots i_n}
    =
    \sum_{\zeta_1,\dots,\zeta_{n-1}}
    A^{[1]}_{i_1,\zeta_1}
    A^{[2]}_{\zeta_1,i_2,\zeta_2}
    \cdots
    A^{[n]}_{\zeta_{n-1},i_n}.
\end{equation*}
Here $i_k$ are physical indices of local dimension $2$, and $\zeta_k$ are bond indices with dimension $\chi_k$. 
The bond dimensions determine both the storage cost and the representational complexity across bipartitions, and states that admit low bond dimensions $\chi_k$ possess efficient MPS representations.
Figure~\ref{fig:ttsvd} represents an illustration of the MPS construction using successive SVD across all bipartitions.
\begin{figure}[h]
    \centering
    \includegraphics[width=0.8\linewidth]{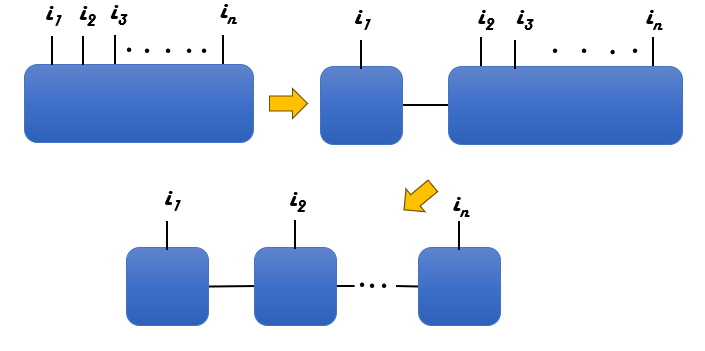}
    \caption{Successive SVD truncations used to obtain the MPS representation of the coefficient tensor. Only the first truncation step is shown explicitly; the downward arrow shows repetition of same decomposition for each remaining sites.}
%Illustration of successive SVD truncation for converting a high dimensional tensor into a matrix product state (MPS).
    \label{fig:ttsvd}
\end{figure}
Similar construction can be applied to an 
operator
$\hat{O}$ acting on $n$ qubits with coefficient {tensor}
$O_{i_1\dots i_n,\, j_1\dots j_n}$,
to obtain a 
matrix product operator (MPO)
\begin{equation*}
    O_{i_1\dots i_n,\, j_1\dots j_n}
    =
    \sum_{\zeta_1,\dots,\zeta_{n-1}}
    W^{[1]}_{i_1 j_1,\zeta_1}
    W^{[2]}_{\zeta_1,i_2 j_2,\zeta_2}
    \cdots
    W^{[n]}_{\zeta_{n-1},i_n j_n}.
\end{equation*}
In practice, directly applying this procedure to the coefficient tensors that are
exponentially large in $n$ 
is computationally prohibitive. 
Instead, the local tensors are built iteratively depending on the specific problem being addressed. For instance, in many-body physics computational approaches such as density matrix renormalization group (DMRG) or time-evolving block decimation (TEBD) are employed.  
In sec.~\ref{sec:ODE} we present a method for the construction of MPOs for differential operators, without explicitly storing the initial {coefficient} 
tensor.

\subsection{Pauli basis expansion}
\label{sec:pauli}

A key step in the efficient construction of an MPO of the differential operator is the Pauli basis expansion
Incidentally, the decomposition is also integral to the implementation of the imaginary time evolution method to find the low excitation states using Suzuki--Trotter approximation.

An $n$ qubit operator $\hat{O}$ can be written
in terms of Pauli strings $P_n^{(\ell)}$:
\begin{equation}
    \hat{O} = \sum_{\ell=1}^{4^n} c_\ell \, P_n^{(\ell)},~\textrm{where}~ P_n^{(\ell)} \in \mathcal{P}_n = \{p_1 \otimes p_2 \otimes \cdots\ \otimes p_n\}.
    \label{op}
\end{equation}
Here $p_k \in \{I, X, Y, Z\}$ represents one of the Pauli matrices and $\mathcal{P}_n$ is the Pauli basis, with
$4^n$ Pauli strings $P_n^{(\ell)}$.
To decompose an operator in the basis $\mathcal{P}_n$, the matrix elements of each $P_n^{(\ell)}$ needs to be efficiently computed. See Sec.~\ref{Appendix:definitions and notation} in Appendix for more details.

Previous work~\cite{Koska2024} 
use sparse arrays to iteratively generate {the matrix elements of Pauli string $P_n$ (dropping the superscript $ \ell$)} in $n$ steps. 
One of the primary
contribution of our work is to formulate
an analytical expression
to obtain the matrix elements of $P_n$ using binary encoding.
To derive the matrix elements for a given Pauli string $P_n$, we define two $n$-bit binary 
{numbers} $\alpha(P_n)$ and $\beta(P_n)$. 
Each row of the matrix $P_n$ has only a single non-zero element, and the binary number
$\alpha(P_n)$ captures the position (column number) of this unique non-zero element, while $\beta(P_n)$ encodes its value. 

\begin{definition} 
% [Structure number $\alpha$]
For any $P_n\in\mathcal{P}_n$, with $P_n=p_1\otimes p_2\otimes \cdots \otimes p_n$, we have a corresponding $n$-bit binary number $\alpha(P_n)$, called the structure number, defined as
\[\alpha (P_n)=\alpha (p_1) \alpha (p_2) \cdots \alpha (p_n) = \alpha_1 \alpha_2 \cdots \alpha_n,\]
with $\alpha (I)=\alpha (Z)=0$ and $\alpha (X)=\alpha (Y)=1$, such that $\alpha_k \in \{0,1\}$.
\end{definition}
\noindent If $j = j_1j_2 \cdots j_n$ is the binary number that denotes the row, then the column number $k=k_1k_2\cdots k_n$ of the non-zero element in the 
$j^{th}$ row
is given by $k=\alpha \oplus j$, where $\oplus$ is bit wise XOR. The observation has been formally stated and proved in the following theorem.

\begin{theorem}[Structure Theorem]
For any $P_n\in\mathcal{P}_n$ (Pauli basis), with $P_n=p_1\otimes p_2\otimes \cdots \otimes p_n$ and corresponding structure number $\alpha$, 
{the column number} of the non-zero element in the $j^{th}$ row is $\alpha \oplus j.$
\end{theorem}
\begin{proof}
We prove this theorem by induction. 
For $n=1$ we can exhaustively check that the Pauli matrices $\{I,Z,X,Y\}$ follow the structure theorem.
Now, assume that for $n=m$ the theorem 
holds i.e., {the column number} of the non-zero element in the $j^{th}$ row of $P_m$ is given as $\alpha \oplus j \equiv (\alpha_1\oplus j_1) \cdots (\alpha_m \oplus j_m)$.  
Any Pauli string of length $n=m+1$ can be formed by a tensor product of a Pauli matrix $P_1$ with the string 
$P_m$, such that $\alpha(P_{m+1}) = \alpha_0\alpha_1\alpha_2\cdots\alpha_m$.
Depending on $\alpha_0$, we have the following two cases:

\noindent \textbf{Case 1:} Assume $P_1$ has 
a structure number $\alpha_0=0$, such that $P_1 \in \{I,Z\}$. The resulting $m+1$-bit structure number of $P_{m+1}$ is $0\alpha_1 \alpha_2 \cdots \alpha_m$. Now, we can observe that in matrix $P_{m+1}$, the matrix $P_m$ will be placed in the diagonal,
resulting in the block diagonal structure: 
\[ 
\begin{bmatrix}
    c_1P_m & 0 \\ 0 & c_2P_m
\end{bmatrix},~
\text{where $c_1,c_2$ are constants.}
\]
For the upper block, the row numbers are $0j_1 j_2 \cdots j_n$ and 
the {column number} of the non-zero element in $P_{m+1}$ is 
$0(\alpha_1\oplus j_1) \cdots (\alpha_m \oplus j_m) = 0\alpha\oplus0j$. Similarly, for the lower block, the row numbers are $1j_1j_2 \cdots j_n$ and the column number is
$1(\alpha_1\oplus j_1) \cdots (\alpha_m \oplus j_m)=0\alpha\oplus1j$.
Therefore, the position of the non-zero element in the $j'^{th}$ row is $0\alpha \oplus j'$, where $0\alpha$ and $j'$ are $m+1$-bit numbers corresponding to the structure and row number of $P_{m+1}$.

\textbf{Case 2:} 
If $P_1$ has 
a structure number $\alpha_0=1$,
such that $P_1 \in \{X,Y\}$, the resulting $m+1$ bit structure number of $P_{m+1}$ 
is $1\alpha_1 \alpha_2 \cdots \alpha_m$. 
The matrix $P_m$ 
now appears in the anti-diagonal blocks, 
such that the matrix $P_{m+1}$ is
\[ 
\begin{bmatrix}
    0 & c_1P_m \\ c_2P_m & 0
\end{bmatrix},~
\text{where $c_1,c_2$ are constants.}
\]
In the upper half 
the {column numbers} of the non-zero elements 
is 
$1(\alpha_1\oplus j_1) \cdots (\alpha_m \oplus j_m)=1\alpha\oplus0j$ and in the lower half it is 
$0(\alpha_1\oplus j_1) \cdots (\alpha_m \oplus j_m)=1\alpha\oplus1j$.
Therefore, the position of the non-zero element in the $j'^{th}$ row is $1\alpha\oplus j'$, where $1\alpha$ and $j'$ are $m+1$-bit numbers corresponding to the structure and row number of $P_{m+1}$.\\

\noindent Hence, the theorem holds for $n=m+1$.
\end{proof}

The second $n$-bit string is the value number $\beta$, which 
is used to find the value of the non-zero element in the $j^{th}$ row and is defined as:

\begin{definition}
% [Value Number: $\beta$]
For any $P_n\in\mathcal{P}_n,\space P_n = p_1\otimes p_2\otimes \cdots \otimes p_n$, we have a corresponding $n$-bit binary number $\beta (P_n)$, called the value number, 
defined as
\[\beta (P_n)=\beta (p_1) \beta (p_2) \cdots \beta (p_n) = \beta_1 \beta_2  \cdots \beta_n,\]
with $\beta (I)=\beta (X)=0 \space \text{  and  } \space \beta (Z)=\beta (Y)=1$,
such that $\beta_k  \in \{0,1\}$.
\end{definition}
\noindent The following theorem proves how $\beta$ contributes to determining
the value of the non-zero element in the $j^{th}$ row of $P_m$. 

\begin{theorem}[Value Theorem]
For any $P_n\in\mathcal{P}_n$ (Pauli basis), with $P_n = p_1\otimes p_2\otimes \cdots \otimes p_n$ and value number $\beta$, the value of the non-zero element in the $j^{th}$ row is 
\[(-i)^{\operatorname{K}(\alpha \land \beta)} \times (-1)^{\operatorname{\Pi}(\beta\land j)},\]
where $\operatorname{K}$ counts the number of $1s$ in a binary number and $\operatorname{\Pi}$ gives the parity of a binary number.
\end{theorem}
\begin{proof}
First,  
the factor of $(-i)^{\operatorname{K}(\alpha \land \beta)}$ arises in the above expression, because $Y$ is replaced by $Y'=iY$ to get a modified Pauli string $P_n'$, 
such that
$P_n = (-i)^{n_Y}\times P_n'$, where $n_Y=\operatorname{K}(\alpha \land \beta)$ is the number of $Y$ operators in the string. The second term $(-1)^{\operatorname{\Pi}(\beta\land j)}$ gives the value of the non-zero element in $j^{th}$ row of $P_n'$. This can again be proved by induction.

For $n=1$, one can verify that the statement holds for
the matrices $\{I,Z,X,Y'\}$. 
Assume that the statement holds for $n=m$, i.e.
the value of non-zero element in the $j^{th}$ row of $P_m'$ is given as $(-1)^{\operatorname{\Pi}(\beta\land j)}$. 
For $n=m+1$, the modified $P'_{m+1}$ can be formed by a tensor product of 
the matrix $P'_1$ with 
$P'_m$. This gives rise to two cases. 

\textbf{Case 1:} Assume $P'_1$ has value number $\beta_0=0$, thus $P'_1 \in \{I,X\}$. 
The matrix $P'_{m+1}$ has the following structure:
\[ 
\left\{\begin{bmatrix}
    P'_m & 0 \\ 0 & P'_m
\end{bmatrix},
\begin{bmatrix}
    0 & P'_m \\ P'_m & 0
\end{bmatrix}\right\}
\]
The value of non-zero element in a block
is $(-1)^{\operatorname{\Pi}(\beta\land j)}$, which results in
\[
(-1)^{(0\land j_0) \oplus (\beta_1\land j_1) \oplus \cdots \oplus (\beta_m \land j_m)},
\]
with $j_0 = 0$ and $1$
for the upper and lower blocks in both the matrices.
Thus, the value of the non-zero element in $P'_{m+1}$ 
is given by $(-1)^{\operatorname{\Pi}(0\beta\land j')}$, where $j'=j_0j$ is the $m+1$-bit 
row number in the modified 
$P'_{m+1}$, with 
value number 
$\beta(P'_{m+1})=0\beta_1 \beta_2 \cdots \beta_m$, which is consistent with this case.

\textbf{Case 2:} Now, assume $P'_1$ has 
a value number $\beta_0=1$, such that $P'_1 \in \{Z,Y'\}$. 
The resulting matrix $P'_{m+1}$ is
\[ 
\left\{\begin{bmatrix}
    P_m & 0 \\ 0 & -P_m
\end{bmatrix},
\begin{bmatrix}
    0 & P_m \\ -P_m & 0
\end{bmatrix}\right\}
\]
The value of non-zero elements 
in the upper and lower block is given by
\[
\begin{array}{rcl}
(-1)^{0 \oplus (\beta_1\land j_1) \oplus \cdots \oplus (\beta_m \land j_m)} & = & (-1)^{(1\land 0)(\beta_1\land j_1) \oplus \cdots \oplus (\beta_m \land j_m)}  \\
     & = &  (-1)^{\operatorname{\Pi}(1\beta\land 0j)},~\textrm{and} \\
 (-1)^{1 \oplus (\beta_1\land j_1) \oplus \cdots \oplus (\beta_m \land j_m)}    
 & = &  (-1)^{(1\land 1)(\beta_1\land j_1) \oplus \cdots \oplus (\beta_m \land j_m)} \\
     & = & (-1)^{\operatorname{\Pi}(1\beta\land 1j)},
\end{array}
\]
respectively. 
Combining the results for the two
blocks, 
the value of non-zero element in the $j'^{th}$ row of 
$P'_{m+1}$  is given by $(-1)^{\operatorname{\Pi}(1\beta\land j')}$
where $j'=j_0j$.
The resulting $m+1$ bit value number $\beta(P'_{m+1})=1\beta_1 \beta_2 \cdots \beta_m$, which is consistent % 
with this case.\\

\noindent Thus, the theorem holds for $n=m+1$.
\end{proof}

We combine the result of the two theorems to get 
an analytical expression for the matrix elements of the Pauli strings $P_n$, in terms of the 
two $n$-bit binary numbers $\alpha$ and $\beta$:
\begin{equation}
P_{jk} =
\begin{cases}
(-i)^{\operatorname{K}(\alpha \land \beta)}(-1)^{\operatorname{\Pi}(\beta \land j)}, & k = \alpha \oplus j, \\
0, & \text{otherwise}.
\end{cases}
\label{eq:pauli_matrix_representation}
\end{equation}
Here, the row ($j$) and column ($k$) indices are $n$-bit binary numbers.
Using the above expression, the
coefficients $c_\ell$, corresponding to $P_n^{(\ell)}$ in the Pauli-basis expansion of any $n$-qubit operator $\hat{O}$ in Eq.~\eqref{op}, can be derived as (see Sec.~\ref{Appendix:Proofs for some propositions} in the Appendix)
\begin{equation}
    c_\ell=\frac{1}{2^n} \sum_{j=0}^{2^n -1}(-i)^{\operatorname{K}(\alpha \land \beta)} \times (-1)^{\operatorname{\Pi}(\beta\land j)}\times \hat{O}_{\alpha\oplus j,j}.
\label{eq:pauli_coefficient_equation}
\end{equation}
Therefore, the derived 
analytical expression in Eq.~\eqref{eq:pauli_coefficient_equation} yields the Pauli 
decomposition of 
an arbitrary $n$-qubit $(2^n \times 2^n)$ matrix, by reducing the entire computation to a direct summation over $2^n$ terms. This represents a significant simplification over existing methods: the iterative construction of Pauli strings, using $\mathcal{O}(2^{n+1})$ sparse array manipulations~\cite{Koska2024} is 
completely avoided and no binary equations are needed to be solved~\cite{Arseniev2024}. Moreover the present result holds for any $n$-qubit matrix without 
any restrictions in its property or symmetry. 
In the following section, the Pauli-basis expansion is applied
to tridiagonal matrices corresponding to the Sturm--Liouville differential operators.

\subsection{Differential equations and tridiagonal matrix decomposition}
\label{sec:tridiag}

Consider the Sturm--Liouville operator, given by
\begin{equation}
\mathcal{L} = \left[\frac{d}{dx}\left(q_0(x)\frac{d}{dx}\right) + q_1(x) \right] 
\end{equation}
{where $q_0(x)$ and $q_1(x)$ are real valued well behaved functions.} Upon discretization using a symmetric stencil, the differential operator is mapped to a tridiagonal matrix $L$, 
where the diagonal and off-diagonal entries 
have explicit functional representations in terms of $q_0(x)$ and $q_1(x)$. 
For a uniform grid with spacing $h$, the symmetric stencil yields
\begin{equation}
L =
\begin{pmatrix}
a_1 & b_1 & 0 & \cdots & 0 \\
c_1 & a_2 & b_2 & \cdots & 0 \\
0 & c_2 & a_3 & \cdots & 0 \\
\vdots & \vdots & \vdots & \ddots & b_{N-1} \\
0 & 0 & 0 & c_{N-1} & a_N
\end{pmatrix},
\end{equation}
\begin{align}
&a_i = -\frac{1}{h^2}\left[q_0\!\left(x_{i+\frac12}\right)+
q_0\!\left(x_{i-\frac12}\right)
\right] + q_1(x_i), \\
&b_i = c_i = \frac{1}{h^2}\, q_0\!\left(x_{i+\frac12}\right)
\end{align}
The derivation of the above expression is shown in Sec.~\ref{Appendix:Proofs for some propositions} of the Appendix.

It can be proved 
that Pauli strings having non-zero coefficients in the Pauli-basis expansion of tridiagonal matrix follow a specific structure, which can be represented by the structure number $\alpha$ as
$
\label{eq:structure_number}
\alpha =
\underbrace{00\cdots0}_{n-w}
\underbrace{11\cdots1}_{w},
$
where $n$ is the number of qubits and $w$ goes from $0$ to $n$. Also, the value number $\beta$ for a symmetric tridiagonal matrix is such that $\operatorname{K}(\alpha \land \beta)$ is even, which is the case for the discretized differential operator using the symmetric stencil. With this restriction on $\alpha$ and $\beta$, the coefficient computation can be optimized by separating the contributions of diagonal, upper-diagonal, and lower-diagonal terms into three finite sums. The detailed expressions and pseudo code are provided in Sec.~\ref{Appendix:Proofs for some propositions} and ~\ref{Appendix:Algorithms} in the Appendix.

The computational cost of Pauli-basis expansion is primarily determined by the number of Pauli strings having non-zero coefficient. For the tridiagonal matrix considered here, under the imposed constraints on $\alpha$ and $\beta$, the number of non-zero coefficients is $(n+1)2^n$.
\begin{table}[h]
    \centering
    \begin{tabular}{|c |c |c |}
        \hline
        {~~\textbf{Number of}~~} & {~~\textbf{Total time}~~} & 
        {~~\textbf{Time per coefficient}~~} \\
         {\textbf{qubits}}&
         {\textbf{(seconds)}}& {\textbf{(microseconds)}}\\
        \hline
        18 & 0.573 & 2.18 \\
        20 & 1.136 & 1.08 \\
        22 & 5.479 & 1.31 \\
        24 & 24.07 & 1.43 \\
        26 & 122.09 & 1.82 \\
        28 & 570.66 & 2.13 \\
        \hline
    \end{tabular} 
    \caption{Wall-clock time to obtain the Pauli-basis expansion of a Laplacian operator, using a 64~thread, 2.1~GHz processor with 128~GB memory.} 
    \label{tab:qubit_time}
\end{table}
Table~\ref{tab:qubit_time} reports the wall-clock time required to implement the Pauli-basis expansion of the Laplacian operator $(a_i=-2/h^2,\, b_i=c_i=1/h^2)$ for different qubit counts. Since the $n$-qubit Laplacian contains $2^n$ non-zero Pauli strings, the total execution time increases with the system size. We can see that the time per coefficient remains in microseconds for all considered qubit counts. Figure~\ref{fig:time_scaling} shows the wall-clock time scaling of the Pauli-basis expansion as a function of qubit number for the Laplacian (left column) and Legendre (right column) differential operators.
While the runtime scaling is similar across methods for Legendre operator, the key advantage of the present approach is scalability: it continues to operate beyond $n = 15$ qubits, a regime where dense matrix-based methods become memory-infeasible.
\begin{figure}[h]
    \centering \includegraphics[width=1\linewidth]{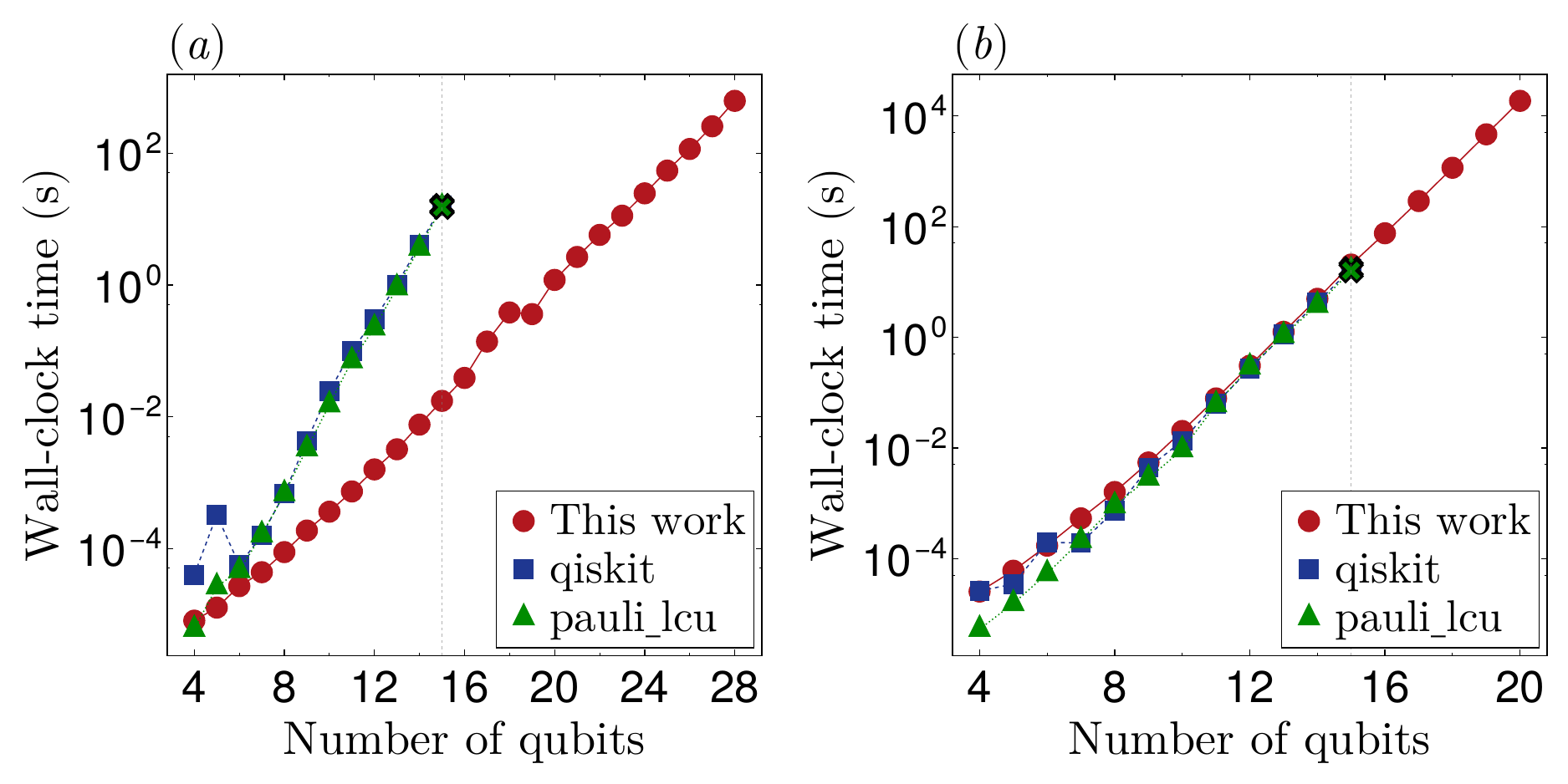}
    \caption{{Wall-clock time scaling of the Pauli-basis expansion as a function of the number of qubits $n$ for the (a) Laplacian and (b) Legendre operators. Red circles (solid line) denote the present approach, navy squares (dashed line) denote the Qiskit and green triangles (dotted line) denote pauli\_lcu~\cite{georges2025paulidecompositionfastwalshhadamard} implementations. The $\times$ marker indicates the point beyond which Qiskit and pauli\_lcu exhausts available memory ($n > 15$), while the present approach continues to function.}}
    \label{fig:time_scaling}
\end{figure}

\subsection{MPO creation and imaginary time evolution}
\label{sec:ODE}

The matrix product operator (MPO) can be constructed by exploiting
the Pauli-basis expansion of the operator 
$L$, thus avoiding the need to explicitly build and store the full coefficient tensor. For example the tridiagonal operators have $(n+1)2^n$ non-zero coefficients in the decomposition, which is considerably lesser than the size of Pauli basis ($4^n$).
Each Pauli string ($P_n$) is a tensor product of local Pauli matrices ($p_k$) and has a simple MPO representation:
\[
\operatorname{MPO}(P_n) = \sum_{\{\zeta_k\}}
W^{[1]}_{\zeta_1}
W^{[2]}_{\zeta_1,\zeta_2}
\cdots
W^{[n]}_{\zeta_{n-1}},~\textrm{where}~
W^{[k]} = p_k. 
\]
Note that the bond indices $\zeta_k$ takes only a single value (say $\zeta_k=1~ \forall k=1,\cdots ,n$). The MPO of the operator $L$ can then be constructed by summing over the Pauli strings MPOs, such that 
\begin{equation}
\operatorname{H_{MPO}}=\operatorname{MPO}(L) = \sum_{\ell} c_{\ell} \, \operatorname{MPO}(P_n^{(\ell)}).
\end{equation}
Now, MPO addition leads to linear increase in the bond dimension. To control the growth of bond dimension, the MPO is compressed via bond truncation after every fixed number of additions, discarding singular values below a prescribed threshold. The truncation error is defined via relative norm change ($\epsilon_{trunc} = || \operatorname{O}-\operatorname{O_{trunc}}||/||\operatorname{O}||$) before and after truncation. The truncation errors in our numerical results are reported in Sec.~\ref{res:err}. 
While the errors grow with system size and become non-negligible on finer grids, they 
remain below $10^{-4}$ for all system sizes considered. Notably, $\operatorname{H_{MPO}}$ is employed solely for coarse-grid eigenvalue extraction, where the truncation errors in operators are less than $10^{-4}$.
For operators with limited correlations, an efficient MPO representation with tractable bond dimension is achievable.

As Sturm--Liouville differential operators are Hermitian, they can be viewed as Hamiltonians of one-dimensional quantum spin systems. Expressing these operators in the Pauli basis establishes a direct connection with quantum lattice models. We consider, $H = \sum_j H_j = \sum_{\ell} c_{\ell} P_n^{(\ell)}$, where $H_j$ is the sum of mutually commuting Pauli strings, scaled with 
respective coefficient.
Such a mapping allows the use of well-studied  
techniques 
such as imaginary time evolution~\cite{McArdle2019, Motta2020} 
to find the ground state and the first few excited states of the Hamiltonian $H$, 
which corresponds to the first few eigenvalues and eigenfunctions of the Sturm--Liouville operator. For implementing imaginary time evolution, we construct the evolution operator using the first-order Suzuki--Trotter approximation. 

Importantly, Pauli strings sharing the same structure number \(\alpha\) and with equal parity of \(\operatorname{K}(\alpha \land \beta)\), commute with each other. This mutual commutativity within each $H_j$ allows $e^{- \Delta t H_j}$ to be computed 
with no Trotter error. 
\begin{equation}
    \operatorname{MPO}(e^{-\Delta t H_j}) = \prod_{P_n^{(\ell)} \in \mathcal{C}_j} e^{-c_{\ell} \, \Delta t \, \operatorname{MPO}(P_n^{(\ell)})}.
\end{equation}
{Here $C_j$ is the set of mutually commuting Pauli strings}. 
The self-inverse property of Pauli strings ensures that exponential of Pauli strings can be analytically computed as:
\begin{equation*}
 e^{a\Delta t  \operatorname{MPO}(P)}=\operatorname{cosh}(a\Delta t) \operatorname{MPO}(I) + \operatorname{sinh}(a\Delta t) \operatorname{MPO}(P).    
\end{equation*}
%%%
%  
Bring together the commuting terms, the first-order Suzuki--Trotter approximation~\cite{Suzuki1976} is used to construct 
the imaginary time evolution operator,
\begin{equation}
\operatorname{E_{MPO}}=\operatorname{MPO}(e^{-H \Delta t}) \approx \prod_{j=1}^{m} \operatorname{MPO}(e^{-H_j \, \Delta t}).
\end{equation}
The construction of $\operatorname{E_{MPO}}$ requires repeated MPO--MPO contractions. However, the bond dimension can be kept low ($\chi=20-50$), with errors close to machine precision for all considered system sizes.
The relative norm change incurred at each contraction is reported in 
Sec.~\ref{res:err}.
In practical implementation of this procedure we find that imaginary time evolution initialized from a random MPS requires a substantially larger number of time steps to converge. To accelerate convergence, we introduce a multistage state-refinement heuristic for constructing an improved initial state. This is described in the following subsection.

%%%%%%%%
\subsection{Multistage state-refinement: a map between Hilbert spaces}
\label{sec:Multistage state refinement}
The coarse-to-fine refinement introduced here is framed as a linear map $\hat{U}:\mathcal{H}_n\to\mathcal{H}_{n+n'}$ rather than a classical array operation for the following reason: since $\hat{U}$ is a linear map, it can be represented as an MPO, so that for large systems, where converting an MPS to a dense vector is infeasible, the interpolation can be carried out entirely within the tensor network framework as an MPO-MPS contraction. For the system sizes benchmarked in this work, direct classical interpolation suffices; the operator formulation ensures the method remains scalable beyond this regime.

Consider a discretization of an 
interval using $2^n$ grid points, with {functions $(f(x) = \sum_{l=0}^{2^n - 1} f(x_l) \, |l_{1} l_{2} \cdots l_n\rangle)$ defined on this interval approximated by vectors in}
\begin{equation*}
    \mathcal{H}_n \cong (\mathbb{C}^2)^{\otimes n}, 
    \qquad \dim(\mathcal{H}_n)=2^n,
\end{equation*}
where the computational basis $\{\ket{l}\}_{l=0}^{2^n-1}$ encodes grid points in binary. Refining the grid by inserting $2^{n'}$ intermediate 
points between each pair of neighbours introduces a local Hilbert space 
$\mathcal{H}_{n'} \cong (\mathbb{C}^2)^{\otimes n'}$, inflating the full 
space to 
$
    \mathcal{H}_{n+n'} \cong \mathcal{H}_n \otimes \mathcal{H}_{n'}.
$

Given a coarse-grid state $\ket{\psi} = \sum_\ell c_\ell \ket{\ell} \in 
\mathcal{H}_n$, we construct the refined amplitudes by linear interpolation 
between neighbouring coefficients,
\begin{equation*}
    c_k^{(\ell)} = \left(1-\frac{k}{2^{n'}-1}\right)c_\ell 
    + \frac{k}{2^{n'}-1}\,c_{\ell+1}, 
    \qquad k=0,\dots,2^{n'}-1.
\end{equation*}
This defines a linear refinement map $\hat{U}:\mathcal{H}_n\to\mathcal{H}_{n+n'}$,
\begin{align}
    \hat{U} 
    &= \sum_{\ell=0}^{2^n-1}\,\sum_{k=0}^{2^{n'}-1}
       \left(1-\frac{k}{2^{n'}-1}\right)
       \ket{k}\ket{\ell}\bra{\ell} \nonumber \\
    &\quad 
    + \sum_{\ell=0}^{2^n-2}\,\sum_{k=0}^{2^{n'}-1}
       \frac{k}{2^{n'}-1}
       \ket{k}\ket{\ell}\bra{\ell+1},
\end{align}
whose action yields the interpolated fine-grid state
\begin{equation}
    \ket{\psi'} 
    = \hat{U}\ket{\psi} 
    = \sum_{j=0}^{2^{n+n'}-1} c'_j\,\ket{j},
\end{equation}
where $c'_j=c_k^{(l)}$ with the composite index $j\equiv(\ell,k)$ running over all $n+n'$ qubits.

This scheme is effective because low-lying eigenstates of differential 
operators are smooth, so interpolated coarse-grid eigenstates lie close to 
the true fine-grid eigenstates. Consequently, imaginary-time evolution on 
the finer grid starts near convergence, substantially reducing computational 
cost. We benchmark this approach on Sturm--Liouville problems arising from 
several PDEs.

\section{Numerical Results}
\label{sec:results}

In this section, we look at some prototypical PDEs that essentially reduce to solving the Sturm--Liouville problem.
The first problem is to solve the 
diffusion equation, which has a well-known analytical solution. 
As such, the fidelity of the solutions obtained from the tensor network method with the analytical results serves as a direct benchmark of the proposed method.
Next, 
the method is applied to anharmonic oscillators, and finally, to one-dimensional systems with quadratic, quartic, and random potentials.
The success of the protocol is analyzed using the efficiency and the relative errors arising from the MPO 
construction, as potentials and bond-dimensions are varied. 
\subsection{The one-dimensional diffusion equation}
\label{res:diff}

Most diffusion processes in nature are described by a simple partial differential equation. In one-dimension, with
Dirichlet boundary conditions and a constant source, the PDE is given by
\begin{equation}
    \frac{\partial u(x,t)}{\partial t} = \kappa \frac{\partial^2 u(x,t)}{\partial x^2} + S;~u(-1,t) = u(1,t) = 0,
\end{equation}
where \(\kappa\) is the diffusion constant and $S$ is a constant source term. The analytical solution is expressed in terms of the eigenfunctions of the Dirichlet Laplacian:
\begin{equation}
\label{Eq:series for diffusion}
    u(t) = \sum_{i} \left[ \langle \phi_i | u_0 \rangle \, e^{-\lambda_i \kappa t } + \frac{S}{\kappa \lambda_i}\left(1 - e^{-\lambda_i \kappa t}\right)\langle \phi_i | 1 \rangle \right] \phi_i,
\end{equation}
where \(\phi_i\) and \(\lambda_i\) are the \(i\)-th eigenvector and eigenvalue of the Laplacian \(L\), respectively and \(\langle \phi_i | u_0 \rangle\) is the overlap of the initial state with the eigenvector \(\phi_i\).

Using imaginary time evolution
with tensor network on $n$ qubits, 
the fidelity of the first few computed eigenstates \(\phi_i^{\mathrm{TN}}\) 
and the corresponding analytical states
\(\phi_i^{\mathrm{exact}}\), is given by
\begin{equation}
    F_n = \frac{|\langle \phi_i^{\mathrm{exact}} | \phi_i^{\mathrm{TN}} \rangle|^2}{\langle \phi_i^{\mathrm{exact}} | \phi_i^{\mathrm{exact}} \rangle \, \langle \phi_i^{\mathrm{TN}} | \phi_i^{\mathrm{TN}} \rangle}.
\end{equation}
The computed values of $F_n$ for different number of qubits $n$, are shown in Fig.~\ref{fig:dirichlet_fidelity}.
We observe that the fidelities of the first 32 eigenstates are above 0.95. 
The fidelity close to unity implies that the numerically obtained eigenfunctions 
closely match the exact analytical solutions.
This confirms that the MPO and MPS representation are very efficient and
accurately approximate the low-lying modes of a differential operator such as the Laplacian. 
These modes are the most physically relevant ones, as they dominate the long-time behaviour of the diffusion process. Importantly, the diffusion equation serves as a strong benchmark for our tensor network based method.

\begin{figure}[h]
\centering
\includegraphics[width=0.9\columnwidth]{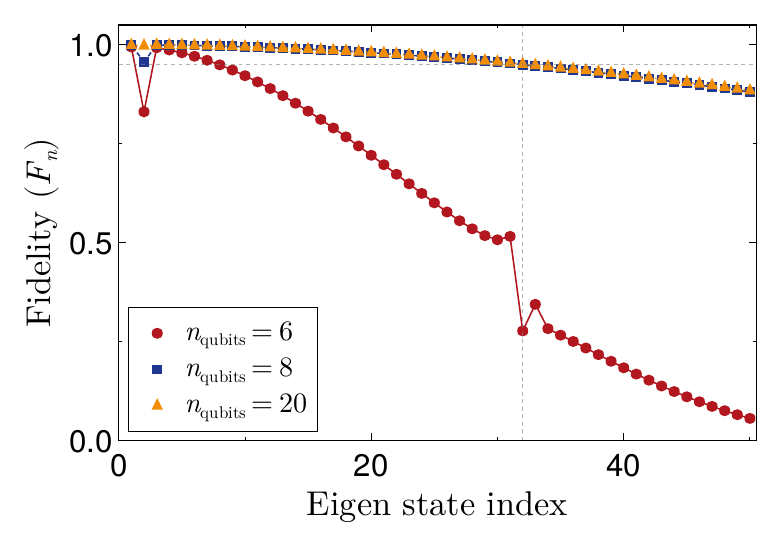}
\caption{Fidelity of first 50 eigenstate of the Dirichlet Laplacian on $[-1,1]$ obtained via multistage imaginary time evolution on 6, 8 and 20 qubits (see legend). For $20$ qubits, first 32 states have $F\geq0.95$ with $F=0.9511$ for state $32$.}
\label{fig:dirichlet_fidelity}
\end{figure}
\begin{figure}[h]
    \centering
    \includegraphics[width=\linewidth]{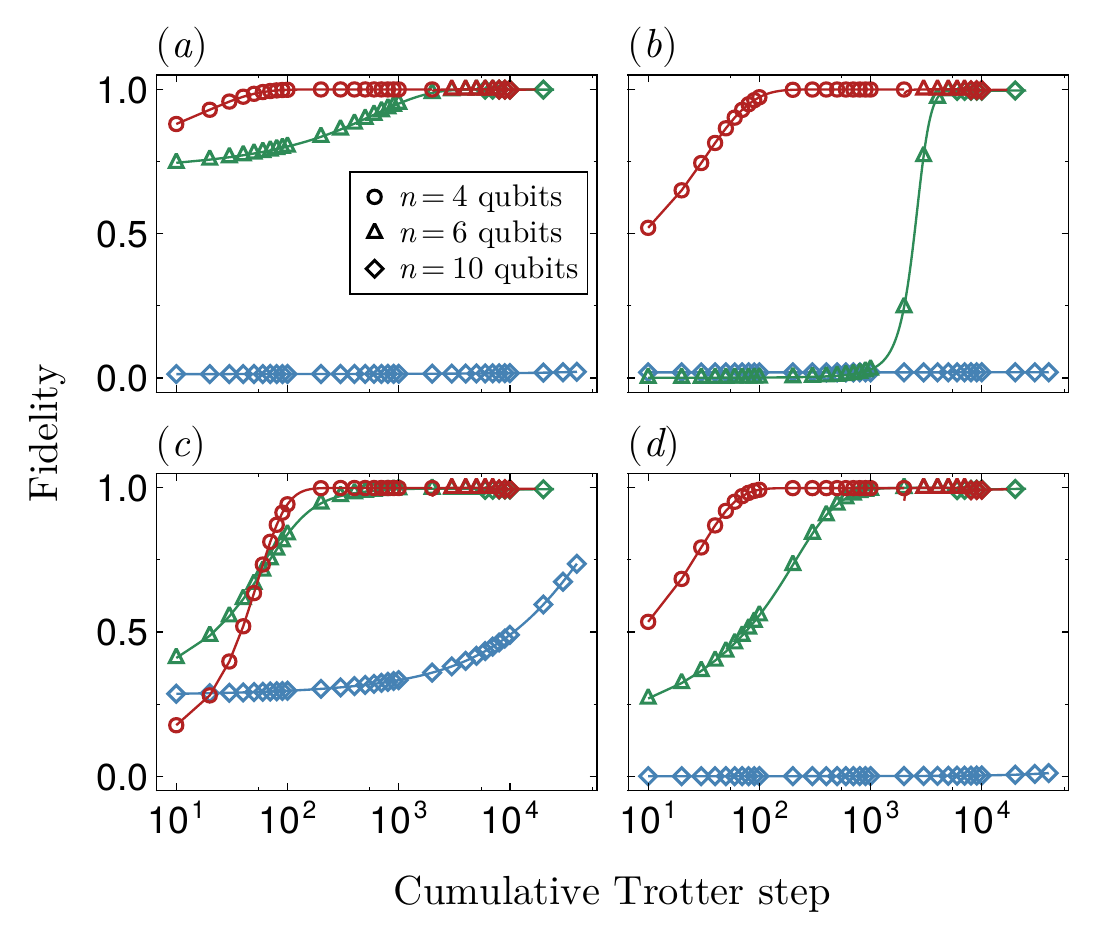}
    \caption{Fidelity vs.\ cumulative number of Trotter step for the (a) ground, (b) first, (c) second, and (d) third excited states of the Laplacian operator. Colors indicate the multistage state-refinement strategy ($n=4\!\to\!6\!\to\!10$ (red), $6\!\to\!10$ (green), or $10$ directly (blue)); markers indicate the qubit count $n$ (see legend, panel a). Trotter step size is kept constant at 0.1.}
    \label{fig:laplacian_fidelity_4panel}
\end{figure}
\begin{table}[h]
    \centering
    \begin{tabular}{|c |c |c |}
        \hline
        \textbf{~Number of~} & \textbf{~Time~} & 
        \textbf{~Number of Pauli Strings~}\\
        \textbf{Qubits} & 
        \textbf{~(seconds)~} & 
        \textbf{in Expansion}\\
        \hline
        8  & 24.02     & 256 \\
        10 & 34.20     & 1,024 \\
        12 & 68.54     & 4,096 \\
        14 & 228.59    & 16,384 \\
        16 & 849.60    & 65,536 \\
        18 & 3,603.44  & 262,144 \\
        20 & 15,899.96 & 1,048,576 \\
        \hline
    \end{tabular}%
    \caption{Wall-clock time for the creation of the MPO representation and the Suzuki--Trotter expansion, along with the corresponding number of Pauli strings, for the discretized Laplacian operator at different qubit counts, measured on a 64-thread, 2.1~GHz processor with 128~GB memory.}
    \label{tab:time_scaling}
\end{table}

Pauli 
decomposition 
of Laplacian differential operator with matrix size larger than 250 million, was efficiently implemented using up to 28 qubits (see Table~\ref{tab:qubit_time}).
For extraction of low-lying modes of the differential operators  using imaginary time evolution, which involves additional construction of Hamiltonian and evolution operator MPOs, the diffusion equation
with Laplacian matrix of size larger than 1 million was solved using about 20 qubits.
The total computational time, 
including the construction of the Hamiltonian MPO, the Suzuki–Trotter 
operator, and the
imaginary time evolution to obtain the ground state is reported in 
Table~\ref{tab:time_scaling}.
Fig.~\ref{fig:laplacian_fidelity_4panel} shows the fidelity between the evolved and exact eigenstates for the Laplacian operator's ground and first three excited states (panels a–d) as a function of cumulative Trotter step. Colors denote the three refinement strategies, and marker shapes indicate the qubit count at each stage. Multistage refinement converges faster than the other strategies, reducing the required number of Trotter steps, while direct evolution at $n=10$ qubits fails to converge even after $40,000$ steps.

\subsection{Two-dimensional anharmonic oscillator}
\label{res:anho}

The stationary Schr\"odinger equation for a two-dimensional anharmonic oscillator is given by
\begin{equation}
\left[
-\frac{\hbar^{2}}{2m}\nabla^{2}
+
V(x,y)
\right]\psi(x,y)
=
E\,\psi(x,y),
\end{equation}
The oscillator is defined on the square domain $\Omega = [-a,a]\times[-a,a]$, where $a=1$\AA, and with Dirichlet boundary conditions. The confining potential is chosen to be separable and anharmonic,
\begin{equation}
V(x,y)=\frac{1}{2} m\omega_x^2 x^2+\frac{1}{2} m\omega_y^2 y^2+
\lambda_x x^4+\lambda_y y^4,
\end{equation}
with parameters
$m\omega_x^2 = 0.1 m \omega_y ^2 =1 ~\mathrm{eV/ \AA^2}$ and
$\lambda_x = 0.2\lambda_y= 1~\mathrm{eV/\AA^4}$.
In this explicitly anharmonic regime, no closed-form analytical solution exists.

The differential operator is discretized on a uniform grid of 
$1024 \times 1024$, with a total of more than a million points. The 
grid spacing is $\Delta x = \Delta y = 2a/N$.
The one-dimensional second-derivative operators are
\begin{equation}
\left(D_{kk}\right)_{ij}
=
\frac{1}{\Delta k^{2}}
\left(
-2\,\delta_{i,j}
+
\delta_{i,j+1}
+
\delta_{i,j-1}
\right),
\end{equation}
where $k \in \{x,y\}$ is the spatial direction and $\delta_{ij}$ is the Kronecker delta function. 
The kinetic and potential energy operators are then given by $T_k = -\frac{\hbar^2}{2m} D_{kk}$ and $\left(V_k\right)_{ij} = \left(\frac{1}{2} m\omega_k^2 k_i^2 + \lambda_k k_i^4 \right)\delta_{ij}$, respectively, and the Hamiltonian is $H_k = T_k+V_k$.

As the potential in the two directions are separable,
the joint eigenstate of the total Hamiltonian is obtained as a tensor product of the two one-dimensional eigenstates.
The lowest eigenstates of the two-dimensional Hamiltonian are shown in Fig.~\ref{fig:ah_eigenstates}.

\begin{figure}[h]
\centering
\includegraphics[width=1\linewidth]{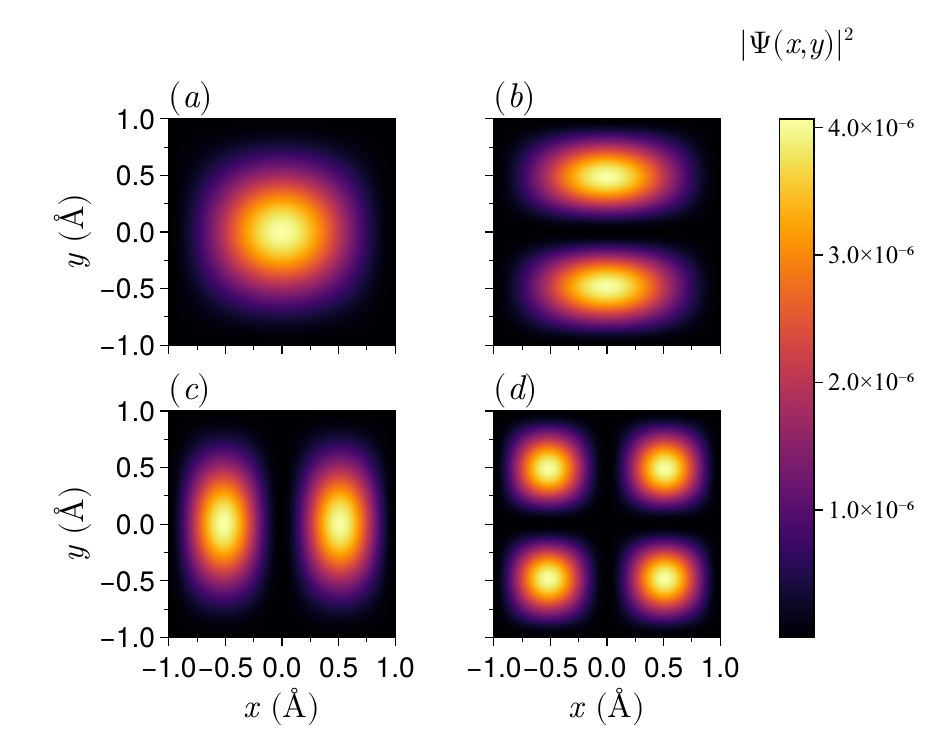}
\caption{Probability density $|\Psi|^2$ of the first four eigenstates of the separable $2D$ anharmonic oscillator confined in a box.$(a)$ Ground state, $(b)$ first, $(c)$ second and $(d)$ third excited states.}
\label{fig:ah_eigenstates}
\end{figure}

In the absence of analytical 
solutions, the eigenstates obtained by the tensor network method can be compared with exact diagonalization results for small systems. For example, we demonstrate results for $n=10$ qubits in each spatial direction. The tensor network method can now be benchmarked in terms of the fidelity $F_n$ with exactly computed eigenstates and relative error $\Delta E$ in  energies.
\begin{figure}[h]
    \centering
    \includegraphics[width=1\linewidth]{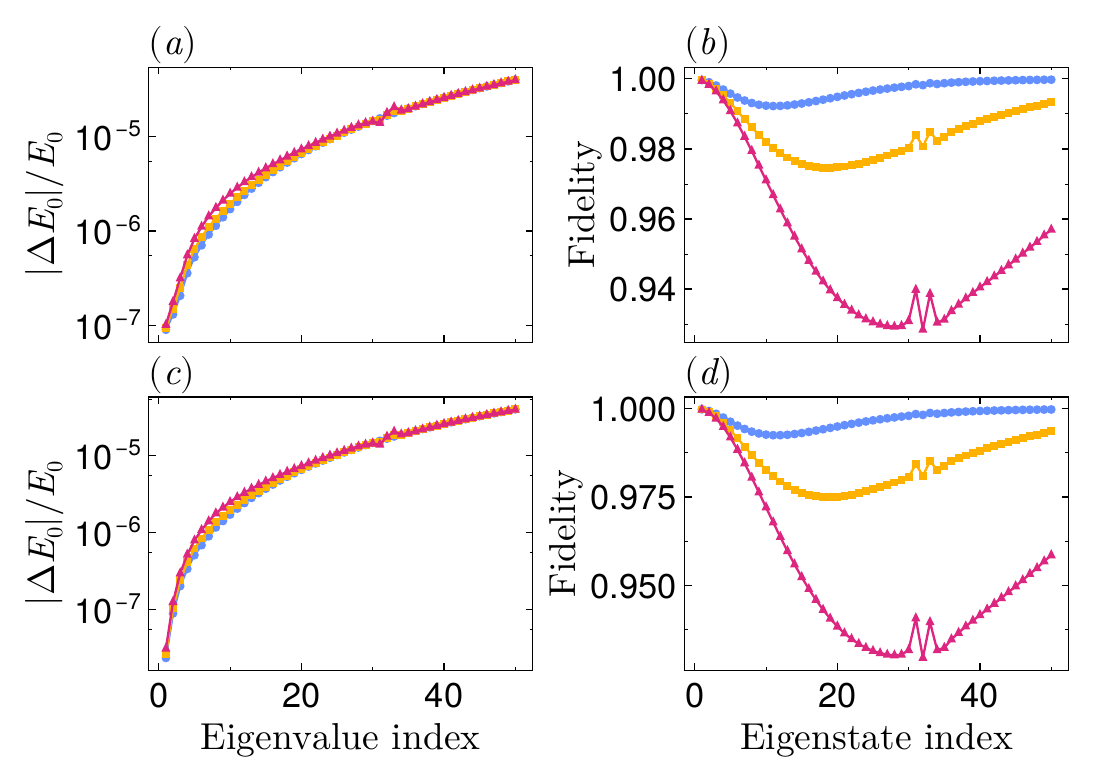}
    \caption{
    Relative eigenvalue error and fidelity of the first $50$ eigenstates of the separable $2D$ anharmonic oscillator. $(a)$ Relative energy error and $(b)$ fidelity for the $x$-Hamiltonian; $(c)$ relative energy error and $(d)$ fidelity for the $y$-Hamiltonian. Blue circles, yellow squares, and pink triangles correspond to $20\,000$, $10\,000$, and $5\,000$ Trotter steps, respectively.
    }
    \label{fig:aho_2d_errors}
\end{figure}

Figure~\ref{fig:aho_2d_errors} presents the relative 
errors in energy $\Delta E$ (left column) and the fidelities (right column), along the $x$ (top row) and $y$ (bottom row) spatial directions.
The results demonstrate strong agreement between the eigenstates and eigenvalues obtained using our method and those from exact diagonalization, with fidelity greater than \(0.99\) and relative error in energy less than \(10^{-4}\). The method can be readily applied to significantly larger grids and operators, with up to a million points in each spatial dimension, using around $20$ qubits.

\subsection{Random potentials in one-dimension}
\label{res:random}

To test the robustness of our 
tensor network method beyond the Laplacian operator and under varying potential strengths, we consider the Sturm--Liouville problem described by the one-dimensional Schr\"odinger equation
\begin{equation}
\left[-\frac{\hbar^2}{2m}\frac{d^2}{dx^2}+V(x)\right]\psi(x)=E\,\psi(x),
\label{Eq:Se}
\end{equation}
defined on the same spatial domain \(x\in[-a,a]\) with Dirichlet boundary conditions. We introduce three representative external potentials,
\begin{align}
V_{\mathrm{quad}}(x) & = g~x^2/a^2, \\
V_{\mathrm{quart}}(x) & = g~x^4/a^4,~\textrm{and}\\
V_{\mathrm{rand}}(x) &= g~W(x),
\end{align}
where $W(x)$ is drawn independently at each grid point from a standard normal distribution.
The discretized differential operator or the
Hamiltonian in Eq.~\eqref{Eq:Se} is discretized using 
a second-order finite-difference scheme, 
similar to the anharmonic oscillator case.
The grid size is taken to be $N=1024$,  
such that a reference solution is obtained via exact diagonalization of discretized Hamiltonian.
The accuracy of the tensor network based method is again benchmarked in terms of the fidelity and the relative eigenvalue error, when compared with the solution obtained using exact diagonalization.

Figure~\ref{fig:fidelity_bond} shows the fidelity error and
relative error in energy eigenvalues, when 
bond dimension \(\chi\) is varied, for the three different potentials.
The strength parameter 
is set to $g=2$~eV (solid line) and $g=10$~eV (dashed line). Figure~\ref{fig:fidelity_strength} show the variation of the errors with 
continuous changes in the potential strength $g$,
while keeping the bond dimension fixed at \(\chi=20\).

\begin{figure}[h]
\centering
\includegraphics[width=1\linewidth]{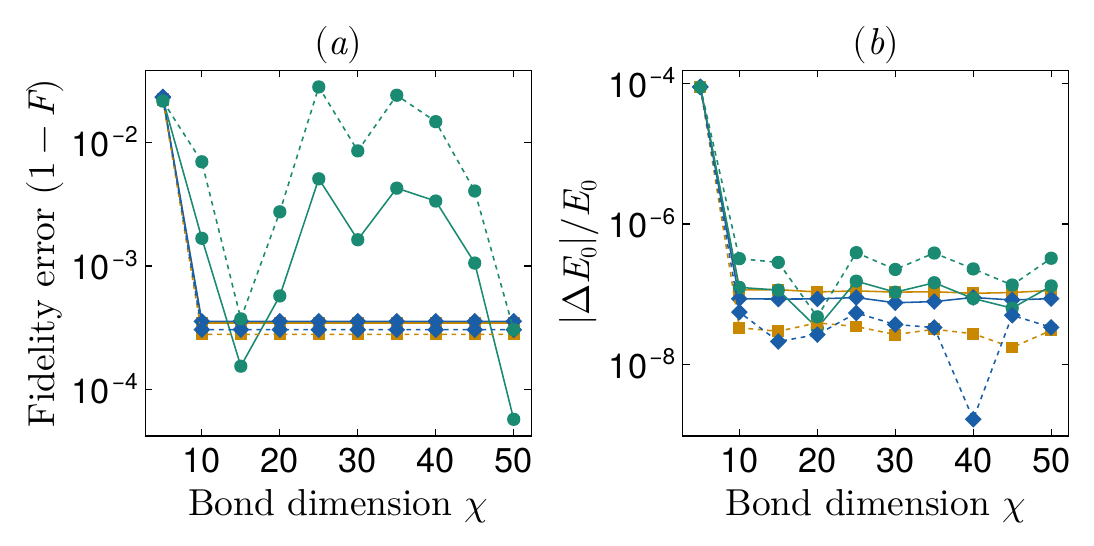}
\caption{(a) Fidelity error $1-F$ and (b) relative energy error of the ground state as a function of bond dimension $\chi$ for quadratic (yellow squares), quartic (blue diamonds), and random (green circles) potentials. Solid and dashed lines correspond to potential strengths of $2$ eV and $10$ eV, respectively.}
\label{fig:fidelity_bond}
\end{figure}
\begin{figure}[h]
\centering
\includegraphics[width=1\linewidth]{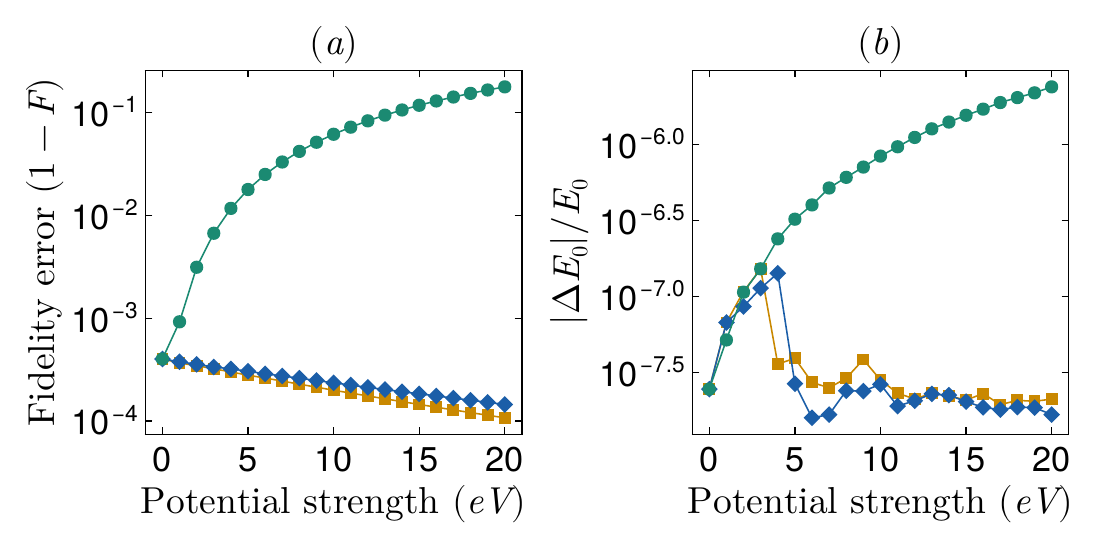}
\caption{(a) Fidelity error $1-F$ and (b) relative energy error of the ground state as a function of potential strength for quadratic (yellow squares), quartic (blue diamonds), and random (green circles) potentials at fixed bond dimension $\chi=20$.}
\label{fig:fidelity_strength}
\end{figure}
%%%
From the error plots of energy and fidelity as functions of bond dimension and potential strength, it is observed that increasing the bond dimension systematically reduces the error. 
Moreover, increasing the potential strength leads to larger errors for random potential. 
It is also noted
that the fidelity of the ground state for the random potential decreases where as that for quadratic or quartic potentials increases slightly. This behaviour indicates that the eigenstates of the random potential become more entangled as the strength of the disorder increases. A similar trend is evident in the relative energy error versus potential strength plot. This suggests that stronger disorder requires a more computationally expensive tensor network representation to accurately capture the corresponding states.

\subsection{Error analysis of MPO constructions}
\label{res:err}

To quantify the numerical accuracy of the constructed matrix product operators (MPOs), we analyse the truncation errors introduced during two stages of the algorithm: (i) construction of the MPO for a differential operator or Hamiltonian based on Pauli basis expansion, and (ii) construction of the evolution MPO using the Suzuki–Trotter approximation.
\begin{figure}[h]
    \centering
    \includegraphics[width=1\linewidth]{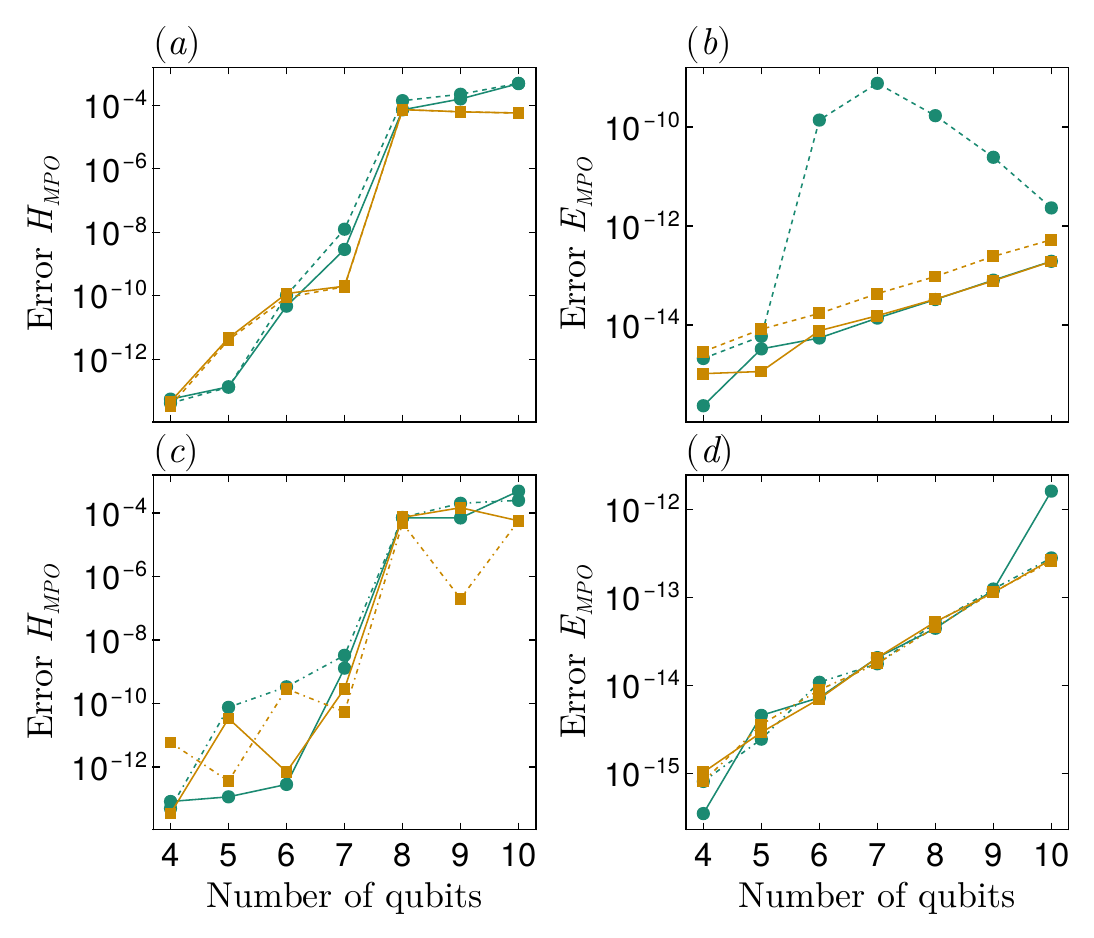}
    \caption{Accumulated bond truncation error in the Hamiltonian MPO (a,c) and the imaginary time evolution operator MPO (b,d) for quadratic (yellow squares) and random (green circles) potentials, as a function of number of qubits. (a,b) At fixed $V=5$ eV, for $\chi=30$ (solid) and $\chi=10$ (dashed). (c,d) At fixed $\chi=20$, for $V=10$ eV (dash--dash) and $V=2$ eV (dash--dot).}
    \label{fig:error_analysis}
\end{figure}

The differential operator or Hamiltonian operator is first
decomposed in the Pauli basis as a weighted sum of Pauli strings. Each Pauli string is converted into an exact MPO representation and the scaled MPOs are successively added to construct the final Hamiltonian MPO. During this addition procedure, the intermediate MPO bond dimension will grow additively. Therefore bond-dimension truncation is applied whenever the bond dimension goes beyond a decided threshold. To estimate the truncation error introduced during this process, we track the relative change in the MPO norm before and after each truncation step. The left panels of Fig.~\ref{fig:error_analysis} show the cumulative error as a function of system size for different Hamiltonians (harmonic and random potential). The results indicate that the accumulated truncation error remains 
low (less than $10^{-4}$), across different system sizes.
This demonstrates that the Pauli-basis MPO construction combined with bond truncation provides an accurate representation of the Hamiltonian.

The time evolution operator is constructed using the Suzuki–Trotter decomposition. Since each term in the decomposition corresponds to the exponential of a Pauli string, its MPO representation can be obtained analytically. The full evolution MPO is then generated through successive 
contractions of MPOs and truncation of the bond dimension.
Directly computing the operator norm error of the resulting MPO is computationally expensive. Instead, we estimate the effective error by probing the action of the approximate MPO on random states. Specifically, we generate $50$ random matrix product states and compute the norm of the resulting states after applying the exact sequence of MPO operations prior to truncation. We then construct the truncated MPO obtained after the MPO 
contractions and apply it to the same set of random states. The relative change in the state norm provides an estimate of the effective operator error, and the reported value corresponds to the maximum error over all sampled states. The right panels of Fig.~\ref{fig:error_analysis} show the accumulated error obtained during successive MPO multiplications in the Suzuki–Trotter construction. The errors remain close to machine precision for all system sizes and parameter regimes considered, indicating that the truncation performed during MPO 
contraction introduces only negligible error in the resulting evolution operator. Due to the exponential growth in the number of Pauli strings with system size, tracking the operator norm becomes computationally expensive for larger numbers of qubits. Therefore, the detailed truncation error analysis is performed only for small system sizes.

\section{Conclusion}
\label{sec:conclusion}

In this work, we developed a tensor network based framework for solving Sturm--Liouville differential equations. 
A central technical contribution of this work is the analytic expression for the matrix element of Pauli strings derived using binary encoding, allowing direct computation of Pauli-basis coefficients without explicit generation of Pauli-string matrix. 
This reduces the space complexity from \(\mathcal{O}(2^{n+1})\) to \(\mathcal{O}(2n)\) in the number of qubits, thereby removing a major computational bottleneck in Pauli-operator based constructions. 
Building on this expansion, we construct the time-evolution MPO via a Suzuki–Trotter decomposition, enabling imaginary-time evolution within the tensor-network framework. To address the growth in required number of Trotter steps with system size, we introduce a multistage state-refinement heuristic, which substantially reduces the number of Trotter steps needed to converge to low-lying eigenstates of the operator for larger systems.

We have benchmarked the method for the diffusion equation and the anharmonic oscillator. For the diffusion equation, imaginary-time evolution on $20$ qubits allowed us to accurately compute the low-lying eigenstates and eigenvalues of the Laplacian operator. The fidelities of the first 32 eigenstates exceed $0.95$.
For the 2-dimensional anharmonic oscillator mapped to two ten-qubit systems, the first $50$ eigenstates were obtained with fidelities greater than \(0.99\), while the relative errors in the corresponding eigenvalues were below \(10^{-4}\). These results indicate that the proposed method can accurately solve differential equations even in cases where analytical solutions are not available. We have also analyzed the truncation errors arising during MPO constructions. The cumulative norm change accumulated during the Pauli operator based Hamiltonian MPO construction remains under $10^{-4}$ across all system sizes considered. Moreover, the errors accumulated during successive MPO contractions in the Suzuki–Trotter expansion remain close to machine precision, confirming the numerical stability of the proposed operator construction. 

We have also tested our method for various potential strengths and bond dimensions. For this analysis, we have considered harmonic, anharmonic and random potential. The first two potentials serve as a baseline in understanding the variations in the random potential as we change the potential strength and the bond dimension. As expected, we have observed that increasing the bond dimension systematically reduces the error. Also, increasing the potential strength causes the fidelity of the eigenstates of the operator with random potential to decrease much faster than those with harmonic and anharmonic potentials. Overall, the results demonstrate that Pauli basis MPO constructions combined with tensor network algorithms provide an efficient method for solving Sturm--Liouville differential equations.

While the overall tensor-network formalism discussed in our work can be implemented for much larger differential operators using more number of qubits in high-performance computers, another important aspect of our results is that the method to 
efficiently obtain expansion coefficients in the Pauli decomposition of a class of operators can be of useful and of independent interest in a wider class of problems. Similarly, the multistage state-refinement technique that accelerates convergence of imaginary-time evolution to operator eigenstates that are sufficiently smooth can be a powerful heuristic in several numerical tools that enable investigations of complex quantum dynamics.

\bibliography{references}

\clearpage

\onecolumngrid
\appendix

\section{Definitions and Notations\label{Appendix:definitions and notation}}

\noindent The single-qubit Pauli matrices,
\[
\mathcal{P_1} = \left\{
\begin{bmatrix} 1 & 0 \\ 0 & 1 \end{bmatrix},\;
\begin{bmatrix} 1 & 0 \\ 0 & -1 \end{bmatrix},\;
\begin{bmatrix} 0 & 1 \\ 1 & 0 \end{bmatrix},\;
\begin{bmatrix} 0 & -i \\ i & 0 \end{bmatrix}
\right\} \equiv \{I, Z, X, Y\},
\]
span the operator space of a single-qubit Hilbert space $\mathcal{H}$.
Their $n$-fold tensor products,
\[
\mathcal{P}_n = \{ p_1 \otimes p_2 \otimes \cdots \otimes p_n \mid p_j \in \mathcal{P_1} \},
\]
define the set of $n$-qubit Pauli strings (denoted $P \in \mathcal{P}_n$), which form a
complete orthogonal basis for the operator space of $\mathcal{H}^{\otimes n}$ under
the Hilbert--Schmidt inner product,
\[
\operatorname{Tr}(P_i^\dagger P_j) = 2^n\,\delta_{ij}, \qquad P_i, P_j \in \mathcal{P}_n.
\]

\section{Proofs for some propositions\label{Appendix:Proofs for some propositions}}

The definitions of the structure number $\alpha$, value number $\beta$, and the analytical expression for the matrix element of Pauli string are given in the main text.

\vspace{2mm}

\noindent\textbf{Proposition 1} (Sparsity of Pauli strings):
\textit{Every Pauli string $P \in \mathcal{P}_n$ has exactly one nonzero entry in each row and each column.}

\begin{proof}
We proceed by induction on the number of qubits $n$. By inspection, each of the four Pauli matrices $I$, $Z$, $X$, $Y$ has exactly one nonzero entry per row and per column.

\noindent Assume every Pauli string in $\mathcal{P}_m$ has exactly one nonzero entry per row and per column. Consider a Pauli string $P \in \mathcal{P}_{m+1}$, which can be written as $P = p \otimes P_m$, where $p \in \mathcal{P}$ and $P_m \in \mathcal{P}_m$. The Kronecker product $p \otimes P_m$ is a $2^{m+1} \times 2^{m+1}$ block matrix whose $(a,b)$ block is $p_{ab}\, P_m$. Since $p$ has exactly one nonzero entry per row and per column, there is exactly one nonzero block $p_{ab}\, P_m$ in each block-row and each block-column. By the inductive hypothesis, $P_m$ itself has exactly one nonzero entry per row and per column, so each nonzero block contributes exactly one nonzero entry to each of its rows and columns. Therefore $P$ has exactly one nonzero entry per row and per column.

By induction, every Pauli string $P \in \mathcal{P}_n$ has exactly one nonzero entry per row and per column for all $n \geq 1$.
\end{proof}

\noindent\textbf{Proposition 2} (Pauli basis coefficients):
\textit{Let $\hat{O}$ be any operator on $\mathcal{H}^{\otimes n}$ expanded in the Pauli basis as $\hat{O} = \sum_\ell c_\ell P^{(\ell)}$, where $P^{(\ell)} \in \mathcal{P}_n$ has associated binary numbers $\alpha$ and $\beta$. Then}
\begin{equation}
c_\ell = \frac{1}{2^n} \sum_{j=0}^{2^n - 1}
(-i)^{\operatorname{K}(\alpha \land \beta)}\,
(-1)^{\operatorname{\Pi}(\beta \land j)}\,
\hat{O}_{\alpha \oplus j,\, j}.
\label{eq_Appendix:coefficient_equation}
\end{equation}

\begin{proof}
Multiplying the expansion $\hat{O} = \sum_\ell c_\ell P^{(\ell)}$ by $P^{(\ell)\dagger}$, taking the trace, and applying orthogonality gives
\[
c_\ell = \frac{1}{2^n}\operatorname{Tr}\!\left(P^{(\ell)\dagger}\hat{O}\right)
=\frac{1}{2^n}\operatorname{Tr}\!\left(P^{(\ell)}\hat{O}\right)
= \frac{1}{2^n}\sum_{j,k} \left(P^{(\ell)}_{jk}\right) \hat{O}_{kj}.
\]
Using $P^\dagger=P\ \ \forall P \in \mathcal{P}_n$. From Eq.~\eqref{eq:pauli_matrix_representation}, the only nonzero entry in row $j$ of $P^{(\ell)}$ occurs at $k = \alpha \oplus j$, so the double sum collapses to a single sum over $j$,
\[
c_\ell = \frac{1}{2^n}\sum_{j=0}^{2^n-1} \left[(-i)^{\operatorname{K}(\alpha \land \beta)}(-1)^{\operatorname{\Pi}(\beta \land j)}\right]\hat{O}_{\alpha\oplus j,\, j}.
\]
\end{proof}

%%%%%%%%%%%%%Proposition for discretization

\noindent\textbf{Proposition 3}(Discretization of Sturm--Liouville operator) Let the Sturm--Liouville operator be:
\begin{equation*}
\mathcal{L}u(x)
=
\frac{d}{dx}\!\left(q_0(x)\frac{du}{dx}\right)
+q_1(x)u(x),
\end{equation*}
where $q_0$ and $q_1$ are well behaved functions. Consider a uniform grid
$
x_i=a+ih,\ i=0,\ldots,N-1,
$
with spacing $h$. Using the centred finite-difference approximation
\begin{equation*}
\left.
\frac{d}{dx}\!\left(q_0(x)\frac{du}{dx}\right)
\right|_{x_i}
=
\frac{1}{h^2}
\left[
q_{i+\frac12}(u_{i+1}-u_i)
-
q_{i-\frac12}(u_i-u_{i-1})
\right]
+\mathcal{O}(h^2),
\end{equation*}
where $q_{i+\frac12}=q_0\!\left(x_i+\frac{h}{2}\right),$ the discrete operator is represented by the tridiagonal matrix
\begin{equation*}
L =
\begin{pmatrix}
a_1 & b_1 & 0 & \cdots & 0 \\
c_1 & a_2 & b_2 & \cdots & 0 \\
0 & c_2 & a_3 & \cdots & 0 \\
\vdots & \vdots & \vdots & \ddots & b_{N-1} \\
0 & 0 & 0 & c_{N-1} & a_N
\end{pmatrix}.
\end{equation*}
with
\begin{align*}
a_i &= -\frac{1}{h^2}\left[
q_0\!\left(x_{i+\frac12}\right)
+
q_0\!\left(x_{i-\frac12}\right)
\right] + q_1(x_i), \\
b_i &= c_i= \frac{1}{h^2}\, q_0\!\left(x_{i+\frac12}\right).
% , \\
% c_i &= -\frac{1}{h^2}\, q_0\!\left(x_{i-\frac12}\right).
\end{align*}
Furthermore, if the same centered stencil is used on both sides of each grid point, then $b_{i+1}=c_i$, and hence $L=L^{T}$ is symmetric.

\begin{proof}
The centered discretization of the flux
$q_0(x)u'(x)$ couples only the neighboring grid values
$u_{i-1}$, $u_i$, and $u_{i+1}$, so every row of the discrete operator has at most three nonzero entries. Expanding the discrete flux gives
\[
(Lu)_i
=
\frac{q_{i+\frac12}}{h^2}u_{i+1}
-
\frac{q_{i-\frac12}+q_{i+\frac12}}{h^2}u_i
+
\frac{q_{i-\frac12}}{h^2}u_{i-1}
+
q_1(x_i)u_i,
\]
which immediately yields the above coefficients. Since the coefficient multiplying $u_{i+1}$ in row $i$ equals the coefficient multiplying $u_i$ in row $i+1$, namely
\[
\frac{q_{i+\frac12}}{h^2},
\]
the upper and lower off-diagonal entries coincide. Therefore the resulting matrix is symmetric and tridiagonal.
\end{proof}

\noindent\textbf{Proposition 4} (Restriction on $\alpha$) Let
\(
\hat{O}=\sum_{\ell} c_{\ell}P^{(\ell)}
\)
be the Pauli basis expansion of an $n$ qubit tridiagonal matrix $\hat{O}$. Then every Pauli string with a non-zero coefficient satisfies
\begin{equation}
\alpha=
\underbrace{00\cdots0}_{n-w}
\underbrace{11\cdots1}_{w},
\qquad
w=0,1,\ldots,n.
\label{Aeq:alpha}
\end{equation}
Equivalently,
\(
\alpha=2^{w}-1.
\)

\begin{proof}
Consider tridiagonal $\hat{O}$ in equation~\ref{eq:pauli_coefficient_equation}. We can observe that the only possible values of $\alpha$ satisfy
\(\alpha=j\oplus j=0,\)
or
\(\alpha=j\oplus(j+1),\). The diagonal case immediately gives $\alpha=0$, corresponding to $w=0$. It therefore remains to determine the possible values of
$\alpha=j\oplus(j+1)$. Suppose that the binary representation of $j$
contains exactly $w-1$ trailing ones. Then $j$ can be written as
\[
j=
b_{n}\cdots b_{w+1}\,
0\,
\underbrace{11\cdots1}_{w-1},
\]
where the bits $b_i$ are arbitrary. Adding one flips the rightmost zero
to one and all trailing ones to zeros, giving 
\(
j+1=
b_{n}\cdots b_{w+1}\,
1\,
\underbrace{00\cdots0}_{w-1}.
\)
Hence,
\(
\alpha=j\oplus(j+1)
=
\underbrace{00\cdots0}_{n-w}
\underbrace{11\cdots1}_{w},
\)
since the higher-order bits remain unchanged while exactly the last
$w$ bits are flipped. Thus every binary string of the form \eqref{Aeq:alpha} occurs,
and no other binary string is possible.
\end{proof}

\noindent\textbf{Proposition 5} (Parity of $\mathrm{K}(\alpha\wedge\beta)$ for real Hermitian operators)
Let $\hat{O}=\hat{O}^{T}$ be real symmetric, expanded as $\hat{O}=\sum_\ell c_\ell P^{(\ell)}$. Then $c_\ell=0$ unless $\mathrm{K}(\alpha\wedge\beta)$ is even.
\begin{proof}
Define $Y'\equiv iY$, real and antisymmetric, while $X,Z,I$ are real symmetric. Writing $P^{(\ell)}=(-i)^{\mathrm{K}(\alpha\wedge\beta)}P'^{(\ell)}$ with $P'^{(\ell)}$ real (each $Y\to Y'$), transposition gives $(P'^{(\ell)})^{T}=(-1)^{\mathrm{K}(\alpha\wedge\beta)}P'^{(\ell)}$. If $\mathrm{K}(\alpha\wedge\beta)$ is odd, $(P'^{(\ell)})^{T}=-P'^{(\ell)}$, so using $\hat{O}^T=\hat{O}$ and cyclicity of the trace,
\begin{equation*}
\operatorname{Tr}(P'^{(\ell)}\hat{O})=\operatorname{Tr}\big[(P'^{(\ell)}\hat{O})^{T}\big]=\operatorname{Tr}\big(\hat{O}(P'^{(\ell)})^{T}\big)
\end{equation*}
\begin{equation*}
=-\operatorname{Tr}(P'^{(\ell)}\hat{O})\ \Rightarrow\ \operatorname{Tr}(P'^{(\ell)}\hat{O})=0.
\end{equation*}
Since $c_\ell\propto\operatorname{Tr}(P'^{(\ell)}\hat{O})$, it vanishes whenever $\mathrm{K}(\alpha\wedge\beta)$ is odd.
\end{proof}

\noindent\textbf{Proposition 6} (Number of Pauli strings in a tridiagonal operator)
Let $\hat{O}$ be a general $n$-qubit tridiagonal operator. The number of Pauli strings $P^{(\ell)}$ with nonzero coefficient is $(n+1)2^{n}$. If, in addition, $\hat{O}$ is symmetric, this number reduces to $n\,2^{n-1}+2^{n}.$

\begin{proof}
By the Restriction on $\alpha$ (proposition 4), $\alpha=2^{w}-1$ for $w=0,1,\ldots,n$, giving $n+1$ allowed values; $\beta$ is unrestricted, giving $2^{n}$ choices each. Hence the total is $(n+1)2^{n}$.

If $\hat{O}$ is symmetric, proposition 5 further requires $\mathrm{K}(\alpha\wedge\beta)$ even. For $\alpha=0$, $\alpha\wedge\beta=0$ for all $\beta$, so all $2^{n}$ values of $\beta$ are allowed. For each of the $n$ nonzero choices $\alpha=2^{w}-1$, $\mathrm{K}(\alpha\wedge\beta)$ depends only on any one bit of $\beta$; fixing the first $n-1$ bits freely ($2^{n-1}$ choices) fixes the parity of the last bit uniquely, leaving exactly $2^{n-1}$ allowed values of $\beta$. 
\end{proof}

\noindent\textbf{Proposition 7} (Pauli strings for a Laplacian operator)
For an $n$-qubit Laplacian tridiagonal operator $\hat O$ (constant diagonal $d$ and constant off-diagonals $u$), the coefficient $c_\ell$ for $\alpha=2^{w}-1$ is nonzero only if $\beta<2^{w}$, $w=0,\dots,n$, giving a total of $2^{n}$ Pauli strings with nonzero coefficient.

\begin{proof}
By Proposition 5, $K(\alpha\wedge\beta)$ is even, so $(-i)^{K(\alpha\wedge\beta)}=s(\alpha,\beta)\in\{\pm1\}$ is independent of $j$ and factors out of Eq.~\eqref{eq_Appendix:coefficient_equation}, leaving $c_\ell\propto\sum_j(-1)^{\Pi(\beta\wedge j)}\hat O_{\alpha\oplus j,j}$.

\emph{Diagonal, $\alpha=0$:} $\hat O_{j,j}=d$ for all $j$, so $c_\ell\propto\sum_{j=0}^{2^n-1}(-1)^{\Pi(\beta\wedge j)}=2^n\delta_{\beta,0}$; only $\beta=0$ survives.

\emph{Off-diagonal, $\alpha=2^{w}-1\neq0$:} take $\alpha\oplus j=j+1$ (the lower-shift case $\alpha\oplus j=j-1$ follows identically by symmetry). Since $\alpha$ flips only the bottom $w$ bits, matching $\alpha\oplus j=j+1$ fixes those bottom $w$ bits of $j$ to the pattern $0\underbrace{1\cdots1}_{w-1}$, while the top $n-w$ bits of $j$ are free; $\hat O_{\alpha\oplus j,j}=u$ for all such $j$. Splitting $\beta=(\beta_{\rm low},\beta_{\rm high})$ accordingly, $\Pi(\beta\wedge j)=\Pi(\beta_{\rm low}\wedge j_{\rm low})\oplus\Pi(\beta_{\rm high}\wedge j_{\rm high})$. Summing over the free bits $j_{\rm high}$ gives $2^{n-w}\delta_{\beta_{\rm high},0}$, so $c_\ell\neq0$ only if $\beta_{\rm high}=0$, i.e. $\beta<2^{w}$, with the bottom $w$ bits of $\beta$ unconstrained. Summing the $2^w$ allowed values of $\beta$ over $w=0,\dots,n-1$, with the boundary term at $w=n$ contributing $\beta=0$,
\begin{equation*}
\sum_{w=0}^{n-1}2^{w}+1=(2^{n}-1)+1=2^{n}.
\end{equation*}
\end{proof}

\section{Pseudo code for Pauli basis coefficient\label{Appendix:Algorithms}}

The coefficients of the Pauli strings in the Pauli-basis expansion of a tridiagonal matrix can be computed using the following subroutine. For computational efficiency, the evaluation is decomposed into three functions, $\operatorname{C0}$, $\operatorname{C1}$, and $\operatorname{C2}$. The function $\operatorname{C0}$ computes the contribution from the diagonal elements, while $\operatorname{C1}$ and $\operatorname{C2}$ compute the contributions from the upper and lower diagonals, respectively.

\begin{algorithm}[H]
\caption{Computation of Pauli basis coefficients for discretized Sturm--Liouville differential operator}
\begin{algorithmic}

\State \textbf{Input}
\Statex $n \gets$ Number of qubits
\Statex $a \gets$ Diagonal function
\Statex $b \gets$ Upper-diagonal function
\Statex $c \gets$ Lower-diagonal function
\Statex $\alpha \gets$ Structure Number
\Statex $\beta \gets$ Value Number
\Comment{$\alpha = 2^{w}-1$ for Tri-diagonal matrix where $w=nX+nY$}

\State \textbf{Initialize} Coeff : Dictionary

\Function{C0}{$n,\alpha,\beta,a$}
  \If{$\alpha = 0$}
    \State $c0=\sum_{j=0}^{2^n-1}
    (-1)^{\text{K}(\beta\land j)}\cdot a[j+1]$
    \State \Return $c0$
  \Else
    \State \Return $0$
  \EndIf
\EndFunction

\Function{C1}{$n,w,\beta,b$}
  \If{$w = 0$}
    \State \Return $0$
  \ElsIf{$0 < w < n$}
    \State $c1=\sum_{j=0}^{2^{n-w}-1}
    (-1)^{\text{K}(\beta\land(2^{w-1}+j\cdot2^w))}
    \cdot b[2^{w-1}+j\cdot2^w]$%b[2^{w-1}+j\cdot2^w+1]
    \State \Return $c1$
  \ElsIf{$w = n$}
    \State \Return
    $(-1)^{\text{K}(\beta\land2^{w-1})}\cdot b[2^{w-1}]$%b[2^{w-1}+1]
  \EndIf
\EndFunction

\Function{C2}{$n,w,\beta,c$}
  \If{$w = 0$}
    \State \Return $0$
  \ElsIf{$0 < w < n$}
    \State $c2=\sum_{j=0}^{2^{n-w}-1}
    (-1)^{\text{K}(\beta\land(2^{w-1}+j\cdot2^w-1))}
    \cdot c[2^{w-1}+j\cdot2^w]$
    \State \Return $c2$
  \ElsIf{$w = n$}
    \State \Return
    $(-1)^{\text{K}(\beta\land(2^{w-1}-1))}
    \cdot c[2^{w-1}]$
  \EndIf
\EndFunction

\Statex

\Procedure{Compute Coefficient}{}
\State
$\text{Coeff[name]}
\gets
\dfrac{(- i)^{\;\mathrm{mod}(\mathrm{K}(\alpha\land\beta),4)}}{2^n} \cdot \big(
\Call{C0}{n,\alpha,\beta,a}
+\Call{C1}{n,w,\beta,b}
+\Call{C2}{n,w,\beta,c}
\big)$
\EndProcedure

\end{algorithmic}
\end{algorithm}

\end{document}